\documentclass[twoside,a4paper]{article}

\usepackage{amsmath,graphicx,amssymb,fancyhdr,amsthm,enumerate,textcomp} 
\usepackage[usenames]{color}
\usepackage{enumerate}
\usepackage{amsmath}
\usepackage{amssymb}
\usepackage[small,bf]{caption}
\usepackage{slashed}
\usepackage{centernot}
\usepackage{dsfont}
\usepackage[utf8]{inputenc}
\usepackage{url}
\newtheorem{thm}{Theorem}[section] 
 
\newtheorem{lem}[thm]{Lemma} 
\newtheorem{prop}[thm]{Proposition}
 
\theoremstyle{definition} 

\newtheorem{cond}[thm]{Conditions}  
\theoremstyle{remark}  
  
\def\beq{\begin{eqnarray}}  
\def\eeq{\end{eqnarray}}  
\def\bsp{\begin{split}}  
\def\esp{\end{split}}

\def\d{\mathrm{d}}

\newcommand{\mf}[1]{{\mathfrak #1}}   
   
\newcommand{\mb}[1]{{\mathbb #1}}   
   
\newcommand{\mbold}[1]{\mbox{\boldmath{\ensuremath{#1}}}}

\def \bl {\mbox{{\mbold\ell}}}
\def \bn {\mbox{{\bf n}}}
\def \bm {\mbox{{\bf m}}}

\fancypagestyle{firststyle}
{
   \fancyhf{}
   \fancyhead{}
   
   \fancyhead[R]{KOBE-TH-14-09}
}
   
\begin{document}

\title{\Large\textbf{$\mathcal{I}$-degenerate pseudo-Riemannian metrics}}  
\author{{{Sigbj\o rn Hervik$^\text{\textleaf}$, Anders Haarr$^\text{\textleaf}$ and Kei Yamamoto$^\text{\textmarried}$  }  }
 \vspace{0.3cm} \\     
$^\text{\textleaf}$Faculty of Science and Technology,\\     
 University of Stavanger,\\  N-4036 Stavanger, Norway         
\vspace{0.3cm} \\      
$^\text{\textmarried}$Department of Physics, Kobe University, \\
Kobe, 657-8501, Japan
\vspace{0.3cm} \\
\texttt{sigbjorn.hervik@uis.no}
\\     
\texttt{anders.haarr@uis.no}
\\      
\texttt{kei@phys.sci.kobe-u.ac.jp}
\\ }     
\date{\today}     
\maketitle   
\thispagestyle{firststyle}
\pagestyle{fancy}   
\fancyhead{} 
\fancyhead[EC]{Hervik, Haarr and Yamamoto}   
\fancyhead[EL,OR]{\thepage}   
\fancyhead[OC]{$\mathcal{I}$-degenerate pseudo-Riemannian metrics}   
\fancyfoot{} 

\abstract 

In this paper we study pseudo-Riemannian spaces with a degenerate curvature structure i.e. there exists a continuous family of metrics having identical polynomial curvature invariants. We approach this problem by utilising an idea coming from invariant theory. This involves the existence of a boost which  is assumed to extend to a neighbourhood. This approach proves to be very fruitful: It produces a class of metrics containing all known examples of $\mathcal{I}$-degenerate metrics. To date, only Kundt and Walker metrics have been given, however, our study gives a plethora of examples showing that $\mathcal{I}$-degenerate metrics extend beyond the Kundt and Walker examples. 

The approach also gives a useful criterion for a metric to be $\mathcal{I}$-degenerate. Specifically, we use this to study the subclass of VSI and CSI metrics (i.e., spaces where polynomial curvature invariants are all vanishing or constants, respectively).
\newpage 

\tableofcontents 
\section{Introduction}
In differential geometry and in pseudo-Riemannian geometry, one can form polynomial curvature invariants by taking contractions of the Riemann tensor and its covariant derivatives. For example, the Ricci scalar, and the Kretschmann scalar, $R_{\mu\nu\alpha\beta}R^{\mu\nu\alpha\beta}$, are simple examples of such invariants \cite{invariants}. 
\\
\\
Let ${\mathcal{I}}$ be the set of \emph{all} such polynomial invariants formed by full contractions of the Riemann tensor and its covariant derivatives. This set is finitely generated \cite{Goodman}, and hence, we can assume that $\mathcal{I}$ is finite. If we are given a metric, $g$, we can compute the value of these invariants, ${I}[g]\in \mathcal{I}$. A question is now, to what extent is the value ${I}[g]$ unique? Here, we will discuss \emph{$\mathcal{I}$-degenerate metrics} which are defined to be metrics having a \emph{degenerate curvature stucture} in the sense that there are continuous families $g_{\tau}$ of non-diffeomorphic metrics\footnote{We will use the following terminology: A \emph{diffeomorphism}, $f: M\rightarrow N$ is a smooth bijection between smooth manifolds. If the manifolds, $M$ and $N$, are equipped with metrics, $g$ and $h$ respectively, then if there exists a diffeomorphism $f$ such that the metrics are related via $f^*h=g$,  we will say that the \emph{metrics are diffeomorphic}. If no such $f$ exists, then $h$ and $g$ are non-diffeomorphic metrics. In the differential geometry literature, one often uses the term isometry for such a map, however, we will reserve the word isometry for diffeomorphisms $f:M\rightarrow M$ where $f^*g=g$. } having the same invariants, i.e., $I[g_\tau]$ does not depend on $\tau$ \cite{degen}. 
\\
\\
In the Riemannian case where the metric is positive definite, there are no $\mathcal{I}$-degenerate metrics, implying that the space is completely determined by the value of $I[g]$ \cite{OP}. In the Lorentzian case, the situation is very different, as there is a large family of $\mathcal{I}$-degenerate metrics \cite{degen}. In this case, all known examples belong to the Kundt class and these metrics have been studied in some detail. For example, the VSI metrics (all polynomial curvature invariants vanish), are known to be all Kundt \cite{VSI}. The CSI case (all invariants are constants), has been studied to some extent and the $\mathcal{I}$-degenerate metrics are also believed to be of Kundt class \cite{CSI, kundt}. In general, examples of $\mathcal{I}$-degenerate Lorentzian metrics have only been found in the Kundt class. In other signatures, other possibilities occur, and to date, only Walker examples (in addition to the Kundt metrics) have been given \cite{Walker,VSI1,Alcolado}. In particular, in 4 dimensional neutral space, it was shown that all VSI spaces are of either Kundt or Walker type \cite{VSI2}. 
\\
\\
The question the paper addresses has an obvious interest in the classification problem in differential geometry. Differential geometry is used in a wide area in physics and 
the Lorentzian case is of clear interest in theories of gravity. Other signatures have applications in physics as well; for example, in twistor theory \cite{twistor,Dun}, the metric is a 4 dimensional metric of neutral signature.
\\
\\
In this paper we study pseudo-Riemannian spaces of arbitrary signature with a degenerate curvature structure. We approach this problem in a new way and find new examples of $\mathcal{I}$-degenerate metrics. Indeed, in a systematic study we define a class of metrics which contain all known examples of $\mathcal{I}$-degenerate metrics, including the Kundt and Walker cases. The underlying assumption is motivated by invariant theory which states that a certain boost limit should exist, at least pointwise, for these spaces \cite{VSI2,RS, align, degenInv}. We assume that this boost limit extends to a neighbourhood, thereby constraining the form of the metric. However, this class is sufficiently rich to include \emph{all known examples} of $\mathcal{I}$-degenerate metrics. 
As a by-product of this assumption, we get a way to determine the invariants for such spacetimes using the metric of a simpler space. 
\\
\\
The structure of the paper is as follows: First we look at the form of a pseudo-Riemannian metric under the assumption that there exists  a surface-forming null $k$-form $\mbold{F}$. This yields a generalisation of both the Kundt and Walker spaces and gives a geometric interpretation of the class under consideration. We also identify subclasses of this family of metrics by imposing appropriate conditions on the covariant derivative of $\mbold{F}$, amongst which are the Kundt and Walker classes.  Then we start anew and look at the form of the metric from the point of view of invariant theory. Here we assume the existence of a particular limit and show that the resulting form of the metric is actually a subclass of the previously considered metrics with the closed $k$-form $\mbold{F}$. In other words, using invariant theory we show that there is a subclass of degenerate metrics hiding within the first class. Then we constrain the coefficients of these degenerate metrics by utilising the boost-weight decomposition and the existence of a particular boost limit. Lastly, we discuss the VSI and CSI subclasses.  
\section{Canonical form of the metric}

Here we assume the existence of a $k$-form $\boldsymbol{F} = \boldsymbol{\ell }^1 \wedge \boldsymbol{\ell }^2 \wedge \cdots \wedge \boldsymbol{\ell }^k $, where $\boldsymbol{\ell }^i $'s are all null and mutually orthogonal. We then impose progressively stronger conditions onto derivatives of $\boldsymbol{F}$ to derive a hierarchy of classes. These conditions, envisaged as a generalisation of the Kundt and Walker conditions, allow us to write the metric in a canonical form. In this section no further assumptions will be made, however, in the next section an assumption of degeneracy of the curvature structure of the (general pseudo-Riemannian) metric will yield a subclass of these metrics. Hence, an independent interpretation of these $\mathcal{I}$-degenerate metrics is provided in the current section.
\\
\\
The space-time dimension is $n = 2k + m$ with signature $(k+p , k-p+m) ,\ p \leq m $. We are given $k \leq n/2$ null 1-forms $\boldsymbol{\ell }^{i}$, which are linearly independent and orthogonal. In discussing various geometrical conditions, it is convenient to work in terms of a $k$-form
defined by
\begin{equation}
\boldsymbol{F} = \boldsymbol{\ell }^1 \wedge \boldsymbol{\ell }^2 \wedge \cdots \wedge \boldsymbol{\ell }^k \ .
\end{equation}
Now we impose two surface-forming conditions the first of which is given as follows:
\begin{cond}[Primary surface-forming condition]\label{cond:1}
There exists a 1-form $\boldsymbol{\Sigma }$ such that
\begin{equation}
d\boldsymbol{F} = \boldsymbol{\Sigma } \wedge \boldsymbol{F} \ .
\end{equation}
\end{cond}
Let us introduce a distribution $\mathcal{D}$, which we shall call the orthogonal complement of $\boldsymbol{F}$, defined by
\begin{equation}
\mathcal{D}_p = \{ \boldsymbol{v} \in T_p \mathcal{M}\  | \ \langle \boldsymbol{\ell }^i , \boldsymbol{v} \rangle = 0, i = 1 , \cdots , k \} 
\end{equation}
From the Frobenius theorem, $\mathcal{D}$ is integrable and there exist functions $u^i $ such that
\begin{equation}
\boldsymbol{\ell }^i = \lambda ^i_{\ j} du^j 
\end{equation}
where the matrix $\lambda ^i_{\ j}$ is invertible
and $u^i = {\rm const}$ specifies an integral manifold of $\mathcal{D}$. 
Next, we construct a canonical frame $\{ \boldsymbol{\ell }^i , \boldsymbol{m}^a , \boldsymbol{n}^{\hat{i}}  \} , a = 1 , \cdots , m = n-2k$ that satisfies
\begin{equation}
\begin{split}
& \boldsymbol{g}\left( \boldsymbol{\ell }^i , \boldsymbol{n}^{\hat{j}} \right) = \delta ^{i\hat{j}} \ , \quad \boldsymbol{g} \left( \boldsymbol{m}^a , \boldsymbol{m}^b \right) = \eta ^{ab} \ , \\
& \boldsymbol{g} \left( \boldsymbol{l}^i , \boldsymbol{m}^a \right) = 0 \ , \quad \boldsymbol{g} \left( \boldsymbol{n}^{\hat{i}} , \boldsymbol{n}^{\hat{j}} \right) = 0 \ , \quad \boldsymbol{g} \left( \boldsymbol{n}^{\hat{i}} , \boldsymbol{m}^a \right) = 0 \ .
\end{split}
\end{equation}
$\eta ^{ab}$ is a pseudo-Euclidean metric of signature $(p,m-p)$.
 Denoting the metric-induced isomorphism by $\sharp $ i.e. for an arbitrary vector $\boldsymbol{v}$ and 1-form $\boldsymbol{\omega }$
 \begin{equation}
  \boldsymbol{g} \left( \boldsymbol{\omega }^{\sharp } , \boldsymbol{v} \right) = \langle \boldsymbol{\omega } , \boldsymbol{v} \rangle \ ,
  \end{equation}
  and using the notation
  \begin{equation}
  \boldsymbol{\ell }_{\hat{i}} = \left( \boldsymbol{\ell }^i \right) ^{\sharp } \ , \quad \boldsymbol{n}_i = \left( \boldsymbol{n}^i \right) ^{\sharp } \ , \quad \boldsymbol{m}_a = \eta _{ab} \left( \boldsymbol{m}^b \right) ^{\sharp } \ ,
  \end{equation}
 we note that $\{ \boldsymbol{n}_i , \boldsymbol{m}_a , \boldsymbol{\ell }_{\hat{i}} \}$ is dual to the canonical frame $\{ \boldsymbol{\ell }^i , \boldsymbol{m}^a , \boldsymbol{n}^{\hat{i}}\} $ in the usual sense and 
 \begin{equation}
 \boldsymbol{\ell }_{\hat{i}} \in \mathcal{D} \ , \quad \boldsymbol{m}_a \in \mathcal{D} \ .
 \end{equation}
The vectors $\boldsymbol{\ell }_{\hat{i}}$ constitute a null distribution $\mathcal{D}^{\ast } \subset \mathcal{D}$. 
The following condition then ensures that $\mathcal{D}^{\ast }$ is integrable:
\begin{cond}[Secondary surface-forming condition]\label{cond:2}
In the canonical frame $\{ \boldsymbol{\ell }^i , \boldsymbol{m}^a , \boldsymbol{n}^{\hat{i}}  \}$,
the components of the covariant derivative of $\boldsymbol{F}$ satisfy
\begin{equation}
\nabla _{\hat{i}} F_{a \mu _1 \cdots \mu _{k-1}} = 0 \ , 
\end{equation}
where, and in what follows, the Greek indices run from $1$ to $n$.
\end{cond}
In order to see that this assumption implies the integrability of $\mathcal{D}^{\ast }$,  
let us look at the covariant derivatives of $\boldsymbol{\ell }^i$ along $\mathcal{D}^{\ast }$;
\begin{equation}
\nabla _{\boldsymbol{\ell }_{\hat{i}}} \boldsymbol{\ell }^k = \omega ^k_{\ \hat{i}j} \boldsymbol{\ell }^j + \omega ^k_{\ \hat{i}a}\boldsymbol{m}^a + \omega ^k_{\ \hat{i}\hat{j}} \boldsymbol{n}^{\hat{j}} \ .
\end{equation}
$\omega ^{\lambda }_{\ \mu \nu }$'s are the components of connection 1-forms defined by
\begin{equation}
\omega ^{\lambda }_{\ \mu \nu } \boldsymbol{e}^{\nu } = \nabla _{\boldsymbol{e}_{\mu }}\boldsymbol{e}^{\lambda } \ , \quad \{ \boldsymbol{e}^{\mu } \} = \{ \boldsymbol{\ell }^i , \boldsymbol{m}^a ,\boldsymbol{n}^{\hat{i}} \} \ , \quad \{ \boldsymbol{e}_{\mu } \} = \{ \boldsymbol{\ell }_{\hat{i}} , \boldsymbol{m}_a , \boldsymbol{n}_i \} \ .
\end{equation}
One can compute the connection 1-forms from exterior derivatives of $\boldsymbol{e}^{\mu }$
(see Appendix A), and $\omega ^k_{\ \hat{i} \hat{j}} = 0$ since conditions \ref{cond:1} imply that
there exists a matrix of 1-forms $\boldsymbol{\sigma }^i_{\ j}$  such that $d\boldsymbol{\ell }^i = \boldsymbol{\sigma }^i_{\ j} \wedge \boldsymbol{\ell }^j $. 

Now evaluating the commutator among $\boldsymbol{\ell }_{\hat{i}}$'s;
\begin{eqnarray*}
\left[ \boldsymbol{\ell }_{\hat{i}} , \boldsymbol{\ell }_{\hat{j} } \right] &=& \nabla _{\boldsymbol{\ell }_{\hat{i}}} \boldsymbol{\ell }_{\hat{j}} - \nabla _{\boldsymbol{\ell }_{\hat{j}}} \boldsymbol{\ell }_{\hat{i}} \\
&=& \left( \nabla _{\boldsymbol{\ell }_{\hat{i}}} \boldsymbol{\ell }^j \right) ^{\sharp } - \left( \nabla _{\boldsymbol{\ell }_{\hat{j}}} \boldsymbol{\ell }^i \right) ^{\sharp } \\
&=& \left( \omega ^j_{\ \hat{i}k} - \omega ^i_{\ \hat{j}k} \right) \boldsymbol{\ell }_k + \left( \omega ^j_{\ \hat{i}a} - \omega ^i_{\ \hat{j}a} \right) \boldsymbol{m}_a \ ,
\end{eqnarray*}
one can conclude that it is sufficient for integrability of $\mathcal{D}^{\ast }$ ($\Leftrightarrow \boldsymbol{\ell }_{\hat{i}}$'s being involutive) to have vanishing $\omega ^k_{\ \hat{i}a}$, which
is equivalent to the conditions \ref{cond:2}. The same condition also turns out to be necessary (c.f. Appendix A).
\\
\\
Once the integrability of $\mathcal{D}^{\ast }$ is established, one can choose a coordinate system
$\{ y^I , v^{\hat{i}} \} ,\  I = 1, \cdots , n-k$ so that we can write
\begin{equation}
\boldsymbol{\ell }_{\hat{i}} = \kappa _{\hat{i}}^{\ \hat{j}} \partial _{v^{\hat{j}}} 
\end{equation}
where the matrix $\kappa _{\hat{i}}^{\ \hat{j}}$ is invertible. Since $\mathcal{D}^{\ast } \subset \mathcal{D}$, we have
\begin{equation}
0 = \langle \boldsymbol{\ell }^i , \boldsymbol{\ell }_{\hat{j}} \rangle = \lambda ^i_{\ k}\kappa _{\hat{j}}^{\ \hat{l}}\langle du^k , \partial _{\hat{v}^{\hat{l}}} \rangle \ ,
\end{equation}
hence
\begin{equation}
0 = \langle du^i , \partial _{v^{\hat{j}}} \rangle = \partial _{v^{\hat{j}}}u^i \quad \Leftrightarrow \quad du^i = f^i_{\ I}(y) dy^I
\end{equation}
where $f^i_{\ I}$'s are functions of $y^I$'s only.
Therefore, one can find a coordinate transformation $u^i (y) , x^a (y), \ a = 1, \cdots , m = n-2k$ independently of $v^{\hat{i}}$, so that $\{ u^i , x^a , v^{\hat{i}} \} $ form a local coordinate system of the entire manifold, with $\{ x^a , v^{\hat{i}} \} $ spanning $\mathcal{D}$.
By construction, it is clear that 
\begin{equation}
g \left( \partial _{v^{\hat{i}}} , \partial _{v^{\hat{j}}} \right) = 0 \ , \quad g\left( \partial _{x^a} , \partial _{v^{\hat{i}}} \right) = 0 \ .
\end{equation}
This corresponds to the following form of the metric: 
\begin{equation}
g_{\mu \nu } = \left( \begin{array}{ccc}
				A_{ij} & B_{ib} & a_{i\hat{j}} \\
				B^t_{aj} & g_{ab} & 0 \\
				a^t_{\hat{i}j} & 0 & 0 \\
				\end{array} \right) \ .
\end{equation}
The line element is then \footnote{We will use hats on the index $i$ when it is essential to distinguish between $u^i$ and $v^{\hat{i}}$. However, when there is no fear of confusion, these will be omitted.}: 
\begin{equation}
ds^2 = 2d u^{i } \left( a_{ij}d v^{j } + A_{ij}d u^{j } + B_{i a} d x^a \right) + g_{ab} d x^a d x^b \ .
\end{equation}
The connection 1-forms for this metric can be found in Appendix A. For later reference, it is useful to consider a class of transformations leaving this form of the metric invariant. Consider the transformation:
\beq
(\tilde{u}^i, \tilde{x}^a, \tilde{v}^j)=(u^i,x^a,f^j( u^n,x^b,v^m)) \ ,
\label{v-trafo}\eeq
for functions $f^j$. We note that
\[ d\tilde{v}^j=(\partial_{u^n}f^j) du^n+(\partial_{x^b}f^j) dx^b +(\partial_{v^m}f^j)dv^m.\]
This may be used to simplify the metric functions $a_{ij}$. Especially, if, for a fixed $i$
\beq 
(\partial_{v^n}a_{im})dv^n\wedge dv^m =0 \ , 
\label{a-condition}\eeq
(which means that $\d (a_{im}dv^m) =0$ as functions of $v^n$ ), then we can use the transformation eq.(\ref{v-trafo}) to bring $ a_{ij}\d v^j \mapsto \delta_{ij}\d v^j$ (fixed $i$). 

\subsection{Important subclasses of metrics} 
\label{sect:subclasses}
Let us point out some important subclasses of these metrics given in terms of the $k$-form $\boldsymbol{F}$. This form is null and surface-forming in the sense made accurate by the conditions \ref{cond:1} and \ref{cond:2}. We will assume the following classes to be of increasing speciality, i.e. we assume class $N+1$ is a subclass of class $N$ etc. 

\paragraph{Type I: "Shear-free and expansion-free"}
As we will see later, all of the examples of $\mathcal{I}$-degenerate metrics given in this article belong to this class. In this class the transverse metric $g_{ab}$ is independent of the $v^j$'s in the coordinate basis. This is equivalent to requiring that $F$ obeys:
\[ \nabla_aF_{b\mu_1 \cdots \mu_{k-1}}=0\ , \]
in the canonical basis $\{ \boldsymbol{\ell }^i , \boldsymbol{m}^a , \boldsymbol{n}^{\hat{i}} \}$ constructed above (not in the coordinate basis where $g_{\mu \nu }$ is written down). 
\\
\\
We note that in the special case where $k=1$, i.e., $\boldsymbol{F}$ is a 1-form, then this case reduces to the {\bf Kundt class} \cite{kundt}. Furthermore, for the Kundt metrics, condition II below is automatically satisfied.

\paragraph{Type II: The matrix $a_{i\hat{j} }=\delta_{i\hat{j}}$.}
If the matrix $a_{i\hat{j}}$ obeys equation (\ref{a-condition}), for all $i$, we can use the transformation, eq.(\ref{v-trafo}), to set $a_{i\hat{j} }=\delta_{i\hat{j}}$. This amounts to require that $\mbold{F}$ obeys (again in the canonical frame): 
\[ \nabla_{\hat{j}}F_{\hat{i}\mu_1 \cdots \mu_{k-1}}=0. \]

\paragraph{Type III: The Walker case}
The Walker class is defined as the case when $\mbold{F}$ is invariant \cite{Walker}: 
\[ \nabla_\mu \boldsymbol{F} = k_\mu \boldsymbol{F},\]
for a vector $k$. This condition alone encompasses the above conditions, I and II. In terms of the metric components this means that the functions $B_{ia}$ do not depend on the $v^{\hat{i}}$'s. 

\paragraph{Type IV: $\mbold{F}$ is covariantly constant}
This case is defined by:
\[ \nabla_\mu \boldsymbol{F}=0, \]
and implies conditions I - III are fulfilled. Interestingly, this implies that the null-form fulfills Killing-Yano equations, but it corresponds to the degenerate case where the Killing-Yano tensor is null.

\paragraph{Type V: All $du^i$ are covariantly constant.} This amounts to the case where there is no $v^{\hat{i}}$-dependence in any of the metric functions which implies that the vectors $\partial_{v^{\hat{i}}}$ are Killing vectors.  Thus in this case the space-time posseses $k$ null Killing vectors. 
\\
\\
In the examples given later, metrics from all of these categories I - V can be found. 

\section{Invariant theory and $\mathcal{I}$-degenerate metrics}
We will now review the boost-decomposition method as in \cite{VSI1,VSI2,bw} and introduce the $\textbf{S}_i$ and $\textbf{N}_i$ properties of a tensor. Utilising the ideas from invariant theory and degenerate tensors, we will, under the assumption that the metric has a similar well-defined limit, reach a class of spaces which are degenerate in the sense that their curvature structures are degenerate. 
\subsection{Boost-weight decomposition}

If we have a pseudo-Riemannian manifold with dimension $(2k + m)$ 
and signature $(k, k + m)$, we can choose a suitable null frame so that the metric can be written:
\begin{align}
\label{metric}
ds^2 = 2( \boldsymbol{\ell}^1 \boldsymbol{n}^1 + ... + \boldsymbol{\ell}^i \boldsymbol{n}^i + ... + \boldsymbol{\ell}^k \boldsymbol{n}^k) + \delta_{ab} \boldsymbol{m}^a \boldsymbol{m}^b
\end{align}
where $a,b=1,...,m$.
\\
\\
 First we choose a real null frame so that we can write down the metric as (\ref{metric}). We then look at the $k$ independent boosts that form an Abelian subgroup of the group $SO(k,k + m)$:
\begin{align}
\label{boost}
(\bl^1, \bn^1) &\rightarrow (e^{\lambda_1} \bl^1, e^{-\lambda_1} \bn^1) \nonumber \\
(\bl^2, \bn^2) &\rightarrow (e^{\lambda_2} \bl^2, e^{-\lambda_2} \bn^2) \nonumber \\
&.  \nonumber \\
&.  \nonumber \\
(\bl^k, \bn^k) &\rightarrow (e^{\lambda_k} \bl^k, e^{-\lambda_k} \bn^k)
\end{align}
where the $\lambda_i$'s are real constants.
This is now a pointwise action on $T_pM$. For a tensor $T$, we now introduce \textit{boost weights}, $\boldsymbol{b}$ $\in$ $\mathds{Z}^k$ in the following manner:
We look at, $T_{\mu_1...\mu_n}$, an arbitrary component of $T$ with respect to the frame given in (\ref{metric}).
\begin{enumerate}
\item Consider a boost given in eq.(\ref{boost}). Then it will transform as $T_{\mu_1...\mu_n} \rightarrow e^{(b_1 \lambda_1 + b_2 \lambda_2 + ... + b_k \lambda_k)}T_{\mu_1...\mu_n}$ for some integers $b_1,...,b_k$. 
\item Then $T_{\mu_1...\mu_n}$ is of boost weight $\boldsymbol{b} = (b_1,b_2,...,b_k)$.
\end{enumerate}
Now one can decompose a tensor into boost weights accordingly:
\begin{align}
T = \sum_{\boldsymbol{b} \in \mathds{Z}^k} (T)_{\boldsymbol{b}}
\label{eq:T_b}\end{align}
where $(T)_{\boldsymbol{b}}$ means the projection onto the subspace (of components) of boost weight $\boldsymbol{b}$. 
\\
\\
If we take the tensor product of two tensors $T$ and $S$, the boost weights obey the following additive rule:
\begin{align}
(T\otimes S)_{\boldsymbol{b}}~=~\sum_{\hat{\boldsymbol{b}} + \bar{\boldsymbol{b}} = \boldsymbol{b}} (T)_{\bar{\boldsymbol{b}}} \otimes (S)_{\hat{\boldsymbol{b}}}
\end{align}
\subsection{The $\textbf{S}_i$- and $\textbf{N}$-properties of a tensor}
We first look at a tensor $T$ and define some conditions its components may fulfill:
\\
\\
\textbf{Definition 1}
\begin{align}
\text{B1)~~} (T)_{\boldsymbol{b}} &= 0, \text{for all~} \boldsymbol{b} = (b_1,b_2,...,b_k), ~b_1 > 0 \nonumber \\
\text{B2)~~} (T)_{\boldsymbol{b}} &= 0, \text{for all~} \boldsymbol{b} = (0~,b_2,...,b_k),~ b_2 > 0 \nonumber \\
\text{B3)~~} (T)_{\boldsymbol{b}} &= 0, \text{for all~} \boldsymbol{b} = (0~,0~,...,b_k),~ b_3 > 0 \nonumber \\
&. \nonumber \\
&. \nonumber \\
\text{B$k$)~~} (T)_{\boldsymbol{b}} &= 0, \text{for all~} \boldsymbol{b} = (0,0,...,0,b_k), ~b_k > 0 \nonumber
\end{align}
\textit{A tensor $T$ possesses the $\textbf{S}_1$ property if there exists a null frame such that condition B1) is satisfied. Furthermore the tensor possesses the $\textbf{S}_i$ property if there exists a null frame such that conditions  B1)-B$i$) are fulfilled.}
\\
\\
\textbf{Definition 2}
\\
\textit{A tensor $T$ possesses the $\textbf{N}$ property if there exists a null frame such that conditions B1)-B$k$) are fulfilled and:}
\begin{align}
(T)_{\boldsymbol{b}} = 0, \text{for~} \boldsymbol{b} = (0,0,...0,0). \nonumber
\end{align}
These conditions can be extended \cite{VSI1} in the following manner:
\\
\\
Consider a tensor $T$, which does not necessarily have any of the $\textbf{S}_i$ properties defined above. Since the boost weights $\boldsymbol{b}$ $\in \mathds{Z}^k \subset \mathds{R}^k$, we can utilise a linear transformation that maps the boost weight on a lattice in $\mathds{R}^k$. More precisely, the transformation $G \in GL(k)$ is a map:
\begin{align}
G:\mathds{Z}^k \rightarrow \Gamma 
\end{align}
where $\Gamma$ is a lattice in $\mathds{R}^k$. Now, if there exist a $G \in GL(k)$ such that after having transformed the boost weights to $G \boldsymbol{b}$, the tensor $T$ now satisfies some of the properties above, we say that $T$ possesses the $\textbf{S}_i^G$- or $\textbf{N}^G$ property. If we have two tensors $T$ and $S$, both possessing the $\textbf{S}_i^G$-property, \textit{with the same $G$}, we can form the tensor product:
\begin{align}
(T\otimes S)_{G\boldsymbol{b}}~=~\sum_{G\hat{\boldsymbol{b}} + G\bar{\boldsymbol{b}} = G\boldsymbol{b}} (T)_{G\bar{\boldsymbol{b}}} \otimes (S)_{G\hat{\boldsymbol{b}}}.
\end{align}
So the tensor product also has the $\textbf{S}_i^G$-property. Note also that if $G = I$ then the  $\textbf{S}_i^G$-property reduces to the $\textbf{S}_i$-property.
\\
\\
The role of these properties can be given in terms of the following result \cite{VSI2}: 
\begin{thm}\label{mainthm}
A tensor $T$ is not characterised by its invariants if and only if it possesses (at least) the ${\bf S}_1^G$-property.
\end{thm}
The crucial point in the proof of this is the existence of a boost, $B_\tau$, of the form eq.(\ref{boost}), so that the components of $T$ under the action of the boost has a well-defined limit $\tau \rightarrow \infty$. Recalling some of the proof, 
 considering the tensor $T$  not characterised by the invariants implies the existence of a ${\mathcal X}$ in the Lie algebra of the boosts so that \cite{RS}
\beq
\lim_{\tau\rightarrow \infty}\exp(\tau{\mathcal X})(T) = T_\infty,
\label{eq:limit}\eeq
which is finite.  Let $\tilde{\mbold b}$ be the boost that represents ${\mathcal X}$. Then, if $(T)_{\mbold b}\neq 0$, we get the requirement  $\tilde{\mbold b}\cdot{\mbold b }\leq 0$. In particular, 
\beq
\exp(\tau\tilde{\mbold b}\cdot{\mbold b })(T)_{\mbold b} \longrightarrow  \begin{cases}\quad  (T)_{\mbold b}, \quad& \tilde{\mbold b}\cdot{\mbold b }=0, \\
 \quad 0, \quad  & \tilde{\mbold b}\cdot{\mbold b }<0,
\end{cases}
\eeq
all other $(T)_{\mbold b}$ must be zero (or else the limit will not exist):
\beq 
\label{eq:posbw=0}
(T)_{\mbold b}=0, \quad \tilde{\mbold b}\cdot{\mbold b}>0.\eeq
This implies the ${\bf S}_1$-property. 

Henceforth, we will call the boost that generates the (pointwise) limit, eq. (\ref{eq:limit}), for the \emph{boost vector} and denote it ${\mbold b}$.

\subsection{The $\mathcal{I}$-degenerate metrics}
\label{sect:deg}
We will first prove a result which is useful in the understanding of the relation between different metrics with the same invariants, and to understand how the limit, eq.(\ref{eq:limit}), from the invariant theory point of view, can be achieved. 
\begin{lem}\label{Lemma}
Consider a null-frame $\{\bl^i,\bn^i,{\bm}^a\}$ at a point $p$. Given also a boost of the frame as follows at $p$: 
\beq \label{eq:frameboost}
\{\bl^i,\bn^i,\bm^a\}\mapsto \{e^{\lambda_i\tau}\bl^i,e^{-\lambda_i\tau}\bn^i,\bm^a\}, \quad \tau\in \mb{R}, (\lambda_i)\in \mb{R}^k.\eeq
Then there exist neighbourhoods $U$ and $\widetilde{U}$ of $p$ and coordinate systems $(u^i,v^i,x^a)$ of  $U$ and $(\tilde{u}^i,\tilde{v}^i,\tilde{x}^a)$ of $\widetilde{U}$ where $p$ is the origin of each coordinate system, such that the diffeomorphism $\varphi_\tau:U\rightarrow \widetilde{U}$, given by: 
\beq
(\tilde{u}^i,\tilde{v}^i,\tilde{x}^a)=(e^{\lambda_i\tau}{u}^i,e^{-\lambda_i\tau}{v}^i,{x}^a), 
\label{diffboost}\eeq
induces the boost $(\ref{eq:frameboost})$ at $p$.  Furthermore,  the diffeomorphism $\varphi_\tau$ can be considered as a 1-parameter family of diffeomorphisms generated by the vector field:
\beq
X=\sum_i\lambda_i\left(u^i\frac{\partial}{\partial u^i}-v^i\frac{\partial}{\partial v^i}\right).
\eeq

\end{lem}
\begin{proof}
Choose a sufficiently small neighbourhood, $U$, around the point $p$. The boost eq.(\ref{eq:frameboost}) induces a transformation in the tangent space $T_pM$. Let $\phi: U\rightarrow \mb{R}^{2k+m}$ be a smooth map mapping the one-forms $\bl^i$ onto $du^i$ and $\bn^i$ onto $dv^i$, at $p$. This choice amounts to choosing a coordinate system where the coordinate vectors align with the $\bl^i$'s and $\bn^i$'s at $p$. Such a choice can always be made since $p$ is merely a point. The diffeomorphism (\ref{diffboost}) gives now the desired boost. The vector field $X$ can now be found by differentiation of $\varphi_\tau$ w.r.t. $\tau$. 
\end{proof}

We will use this boost to get a sufficient criterion of $\mathcal{I}$-degenerate metrics. We construct a metric: 
\[ \widetilde{g}_{\tau} =\varphi_\tau^*g.\]
\\
Given $U$, let $\mf{M}$ be the space of smooth metrics on $U$. We will make the following assumption:
\\
\\
{\it There exists a neighbourhood $U$ so that the metric, in the coordinates given, has a finite limit $\lim_{\tau\rightarrow\infty}\varphi_\tau^*g\in \mf{M}$ with a boost with respect to any given point $p\in U$.
}
\\
\\
This assumption places clear constraints on the possible metric. Let us consider these in detail. Now, the boost above is with respect to the origin of the coordinate system. We consider a neighbourhood, $U$, that is sufficiently small to be covered by one coordinate chart. Choose an arbitrary point $p\in U$, which is given by $(u_0^i,v_0^i,x^i_0)$. We shift this point to the origin, by introducing $(\bar{u}^i,\bar{v}^i,\bar{x}^a)=({u}^i-u^i_0,{v}^i-v_0^i,{x}^a-x^a_0)$, and then apply the above boost. The above assumption now implies that the corresponding limit should be finite for all constants $(u_0^i,v_0^i,x^i_0)$ in the neighbourhood $U$.
\\
\\
The metric consists of symmetric components $g_{\mu\nu}(u^i,v^i,x^a)dx^\mu dx^{\nu}$, which, after translation of the point $p$ to the origin of the coordinate system, turn into:
$g_{\mu\nu}(u_0^i+\bar{u}^i,v_0^i+\bar{v}^i,x_0^a+\bar{x}^a)d\bar{x}^\mu d\bar{x}^{\nu}$. With no loss of generality, we can assume the boost given is:
\[ (\bar{u}^i,\bar{v}^i,\bar{x}^a)\mapsto (e^{-\lambda_i\tau}\bar{u}^i,e^{\lambda_i\tau}\bar{v}^i,\bar{x}^a), \qquad \lambda_i>0, \] 
(we include all null-directions having $\lambda_j=0$ in $x^a$). First, we note that the components
\[ d\bar{v}^i d\bar{v}^j, \qquad d\bar{v}^id\bar{x}^a, \]
have to vanish on $U$. This can be seen as follows: if there is a point $p$ on $U$ for which we have
\[ f(p)d\bar{v}^i d\bar{x}^a \neq 0, \qquad \text{then} \qquad f(p)e^{\lambda_i\tau}d\bar{v}^i d\bar{x}^a\rightarrow \infty. \] 
Clearly, the same argument is valid for $d\bar{v}^i d\bar{v}^j$ as well. Next, consider the metric for the transversal space:
\[  g_{ab}d\bar{x}^ad\bar{x}^b.\] 
Since the metric is smooth (as well as the limit),  the partial derivative of $g_{ab}$ w.r.t. $\bar{v}^i$ exists for any $\tau$. Then, considering it as a function of $\bar{v}^i$:
\beq
 g(v^i_0+e^{\lambda_i\tau}\bar{v}^i)_{ab,\bar{v}^i}=g'(v^i_0+e^{\lambda_i\tau}\bar{v}^i)_{ab} e^{\lambda_i\tau}.
\eeq
Consequently, $g'$ has to be zero (same argument as above), and hence, the components $g_{ab}$ do not depend on the $\bar{v}^i$'s, and therefore, $g_{ab}=g_{ab}(\bar{u}^i,\bar{x}^a)$. 

For the components containing one or two $d\bar{u}^i$'s, we note by taking derivatives of various $\bar{v}^j$'s that they must be polynomials in the coordinates $(\bar{v}^i)$, but are arbitrary smooth functions in  $(\bar{u}^i, \bar{x}^a)$. The order of the polynomial (in the $\bar{v}^i$'s) depends on the actual boost, as well as which component we consider. 

In the following we will also for simplicity introduce some notation. Let $P(v_1,v_2,...,v_k)$ be a polynomial in the $v_i$'s with coefficients being arbitrary functions of $(u^i, x^a)$ (Henceforth, we will sometimes let the index of the $v$-coordinates be downstairs due to a more appealing typesetting). Define $\mathcal{P}$ the ring of all such polynomials: 
\[ \mathcal{P}:=\left\{P(v_1,v_2,...,v_k) ~~| ~~ P ~ \text{polynomial,  coefficients depend on}~ (u^i, x^a)\right\} \]
We will define subsets of this set and indicate them with a bracket $[-,..,-]$. The bracket consists of a list of monomials in $v_i$'s and indicates the highest allowable possible power of the $v_i$'s. For example, $[v_1^3,v_2v_3^5]\subset \mathcal{P}$ is the subset including the following powers: $v_1^n, n=0,...,3$, and $v_2^mv_3^q$, $m=0,1$ and $q=0,...,5$. 

We therefore, end up with the following metric (switching back to non-barred coordinates): 
\beq\label{result}
g=2du^i\left(a_{ij}dv^j+A_{ij}du^j+B_{ia}dx^a\right)+g_{ab}dx^adx^b,
\eeq
where $g_{ab}=g_{ab}(u^i,x^a)$ and $a_{ij}$,  $A_{ij}$, and $B_{ia}$, are polynomials belonging to some subset of $\mathcal{P}$, with  arbitrary smooth coefficients in $(u^i,x^a)$. We therefore conclude that, these metrics form a subclass of the metrics of type I considered in section \ref{sect:subclasses}. 

\subsection{The polynomial invariants}
These metrics represent $\mathcal{I}$-degenerate metrics in the sense that many non-diffeomorphic metrics have the same invariants. Indeed, 
\begin{thm}
\label{thm:varphi}
Consider a point $p\in U$, and assume that there is a one-parameter family of diffeomorphisms $\varphi_\tau$ such that $\varphi_\tau(p)=p$ and $\lim_{\tau\rightarrow\infty}\varphi_\tau^*g=g_0\in \mf{M}$. Then any polynomial curvature invariant of $g$ evaluated at $p$ is identical to the corresponding invariant of $g_0$ at $p$. 
\end{thm}
\begin{proof}
At $p$, the diffeomorphism induces a frame transformation on the tangent space. Since any curvature invariant, $I$, does not depend on such a frame choice, the transformation of $I$ under $\varphi_\tau$ is simply:
\[\left.\varphi_{\tau}^*I\right|_{p}=\left.I\right|_{\varphi_\tau(p)}=\left.I\right|_{p}.\]
Since the metric is smooth, including its limit, any derivative of the metric, evaluated at $p$, $\left.\partial^{(n)}g\right|_p$ is well-defined in the limit as well. Consequently, since the invariants are continuous functions in $g$ and its derivatives, the limit implies:
\[ \lim_{\tau\rightarrow \infty}\left.\varphi_{\tau}^*I\right|_{p}=\left.I\right|_{p}=I\left[\lim_{\tau\rightarrow\infty} \left.g\right|_p\right]=\left.I[g_0]\right|_p.
\]
\end{proof}
We also note the following fact about any metric $g_0$ being the limit of such a boost:
\begin{prop}\label{prop:isometry}
Assume that $\lim_{\tau\rightarrow\infty}\varphi_\tau^*g=g_0\in \mf{M}$, where $\varphi_\tau$ is a one-parameter family of diffeomorphisms. Then $\varphi_\tau^*g_0=g_0$ and consequently, $\varphi_\tau$ is an isometry of $g_0$. 
\end{prop}
\begin{proof}
We observe that:
\[ \varphi_\tau^*g_0=\varphi_\tau^*\left(\lim_{\lambda\rightarrow\infty}\varphi_\lambda^*g\right)=
\lim_{\lambda\rightarrow\infty}(\varphi_\lambda\circ\varphi_\tau)^*g=\lim_{\lambda\rightarrow\infty}\varphi_{\lambda+\tau}^*g=\lim_{\lambda\rightarrow\infty}\varphi_\lambda^*g=
g_0.\] 
\end{proof}
This implies that the limiting metric has extra symmetries compared to $g$. We also note that if $g_0$ turns out to be $g$ (perhaps in some disguise), then the $g$ must possess the symmetry $\varphi_\tau$ as well: assume that there is a diffeomorphism $f$ so that $f^*g=g_0$. Then by applying $\varphi_\tau^*$ on each side, we get:
\[ \varphi_\tau^*f^*g=\varphi_{\tau}^*g_0=g_0=f^*g. \] 
Then applying $f^{-1}$ on each side we obtain:
\[ (f^{-1})^*\varphi_\tau^*f^*g=(f\circ\varphi_\tau\circ f^{-1})^*g=(f^{-1})^*f^*g=(f\circ f^{-1})^*g=g.\] 
Hence, $f^{-1}\circ\varphi_\tau\circ f$ is an isometry of $g$. In this case, the $g$ and $g_0$ are diffeomorphic so we might as well just use the (possibly) simpler metric, $g_0$, to represent our space. Clearly, this means that as long as $\varphi_\tau$ (or $f^{-1}\circ\varphi_\tau\circ f$) is not an isometry of $g$, then $g$ and $g_0$ must necessarily be two non-diffeomorphic metrics. In this sense, the existence of $\varphi_\tau$ implies that the metric $g$ (and its curvature structure) is $\mathcal{I}$-\emph{degenerate}.
\\
\\
We now want to take a closer look at the coefficients of the degenerate metrics. Because of the $v^ i$'s transformation properties under the boost the coefficients $a_{ij}, A_{ij}$ and $B_{ia}$ cannot be polynomials in $v^ i$ of an arbitrarily high degree. 

\section{Constraining the coefficients of $\mathcal{I}$-degen\-erate metrics}
Given the null frame and boost in $(\ref{eq:frameboost})$, the form of the $\mathcal{I}$-degen\-erate metrics is:\\
\beq
g=2du^i\left(a_{ij}dv^j+A_{ij}du^j+B_{ia}dx^a\right)+g_{ab}dx^adx^b,
\eeq
where the coefficients $a_{ij}, A_{ij}$ and $B_{ia}$ are polynomials in the $v^i$'s. In order for the limit, $\lim_{\tau\rightarrow\infty}\varphi_\tau^*g\in \mf{M}$ to exist, the coefficients cannot be polynomials of arbitrarily high order in the $v^i$'s. This is because the $v^i$'s transform as $e^{\tau \lambda_i}v^i$ under boosts and the limit might blow up. We now use the boost-weight decomposition on the metric to get a handle on the $v^i$ dependence of the coefficients. 
\\
\\
The vector space decomposition separates the metric into (coordinate)-components that transform as $g_{\mu\nu} \rightarrow e^{(b_1 \lambda_1 + b_2 \lambda_2 + ... + b_k \lambda_k)}g_{\mu\nu}$ under the action of $\varphi_\tau$ (similarly to eq.(\ref{eq:T_b}) but for a coordinate basis). This is therefore a useful point of view when analyzing the limiting behaviour of the components. In addition to the components we must include the behaviour of the differentials $du^idu^j$, $du^idv^j$ and $du^idx^a$ under boosts. This motivates the following definition:
\\
\\
$\boldsymbol{V} = (d_1,d_2,...,d_k) + \bold{v}_{ij}\in\mathbb{Z}^k$, 
\\
\\
The $d_i$'s are non-negative integers corresponding to the subsets $[v_1^{d_1}v_2^{d_2}...v_k^{d_k}]\subset\mathcal{P}$ which we assume the polynomials $a_{ij},A_{ij}$ or $B_{ia}$ belong to. Differentials are accounted for by $\bold{v}_{ij}$ as follows:  For each $i,j$, the $\bold{v}_{ij}$ is a vector given as follows (other components zero):
\begin{itemize}
\item{} For component $a_{ij}$: $-1$ at the $i$'th place (the  $du^i$ differential) and $+1$ at the $j$'th place (the  $dv^j$ differential). If $i=j$ then $\bold{v}_{ij}=0$.
\item{} For component $A_{ij}$: $-1$ at the $i$'th place (the  $du^i$ differential) and $-1$ at the $j$'th place (the  $du^j$ differential). If $i=j$ then $-2$ at the $i$'th place.
\item{} For component $B_{ia}$: $-1$ at the $i$'th place (the  $du^i$ differential). 
\end{itemize}
The vector $\boldsymbol{V}$ is thus different depending on which components of the metric we are looking at. If we pick out the term $2a_{ij} du^i dv^j$ then:
\\
\\
$\boldsymbol{V} = (d_1,d_2,...,d_i-1,...,d_j+1,...,d_k)$.
\\
\\
Demanding that the limit from section \ref{sect:deg} is finite, gives a bound $\boldsymbol{V} \cdot \boldsymbol{b} \leq 0$, where $\boldsymbol{b}$ is the boost representing $\cal{X}$ in eq.(\ref{eq:limit}). Recall that the components of $\boldsymbol{b}$ are zero or positive. The bound $\boldsymbol{V} \cdot \boldsymbol{b} \leq 0$ ensures we get a well behaved limit and the given inequality constrains the maximal degree of the polynomials $[v_1^{d_1}v_2^{d_2}...v_k^{d_k}]$.
\\
\\
Before we show a concrete example we will do two things. Firstly, we are only interested in the maximum degree of the $v^i$'s and so we set $\boldsymbol{V} \cdot \boldsymbol{b} = 0$. Secondly, for a given\footnote{It is the $k$ from the dimension of the manifold: $2k +m$ with signature $(k,k+m)$} $k$ we have a freedom of choice of how we would like to specify the boost-vectors.
\\
\\
We write the boost vector $\boldsymbol{b}$ in the form $\boldsymbol{b}=(n_1, n_2,...,n_k) \in \mathds{Z}^k$. By utilising a linear transformation from $\mathds{Z}^k \rightarrow \Gamma$, where $\Gamma$ is a lattice in $\mathds{R}^k$, we will consider the cases where we can put the following conditions on the entries of the boost vectors\footnote{We should point out that there may be other possibilities for the boost vector ${\mbold b}$; for example, we can generalise point 2. to include $n_i<n_{i+1}\leq 2n_i$. A complete list of all non-equivalent boosts is not given here.}:\begin{enumerate}
\item $n_i \leq n_{i+1}$, 
\item If $n_i < n_{i+1}~\textit{then}~n_{i+1} = 2n_i$
\item For the case $n_i = 0$ then $n_{i+1} = 0$ or 1 
\end{enumerate}
For a given $k$ there are then $2^k$ different boost vectors $\boldsymbol{b}$ in this category. 
\\
\\
\subsubsection*{Examples:}

For $k=2$, the complete set of boost vectors $\boldsymbol{b}$ are $(0,0)$, $(0,1)$, $(1,1)$ and $(1,2)$.\\
\\
For $k=3$, the complete set of boost vectors $\boldsymbol{b}$ in this category are $(0,0,0)$, $(0,0,1)$, $(0,1,1)$, $(0,1,2)$, $(1,1,1)$, $(1,1,2)$, $(1,2,2)$, $(1,2,4)$. \\

Furthermore, in the boost-vectors with a leading zero, e.g., $(0,0,1)$, the zeros indicate that the boost does not involve these directions. This implies that there are no constraints on these variables and no degeneracy in these directions. Consequently, these directions can be included in the transverse space (and hence, generalising the transverse space and allowing it to be pseudo-Riemannian as well). In the case where there are only zeros, $(0,...,0)$, this corresponds to the $\mathcal{I}$-non-degenerate case where no boost exists.

\subsection{A concrete example}

Suppose we set $k = 3$, specify the boost-vector to be $\boldsymbol{b} = (1,2,4)$ and decide to pick out terms in front of $du^1 du^1$. Then $\boldsymbol{V} = (d_1-2,d_2,d_3)$. Writing out the dot product $\boldsymbol{V} \cdot {\mbold b} = 0$, we get:
\[ d_1 + 2d_2 + 4 d_3 = 2\]
From this equality we gather that the $v_i$-dependence in front of $du^1 du^1$ is restricted to the two cases: $(v_1^2,v_2^0,v_3^0)$ or $(v_1^0,v_2^1,v_3^0)$. Hence, $A_{11}\in [v_1^2,v_2]$. If we do this for all the metric components $A_{ij}=A_{ji}$ we get:
\begin{enumerate}[i)]
\item For $A_{11}$: 
$d_1 + 2d_2 + 4d_3 = 2 $
\item For $A_{12}$: 
$d_1 + 2d_2 + 4d_3 = 3 $
\item For $A_{22}$: 
$d_1 + 2d_2 + 4d_3 = 4 $
\item For $A_{13}$: 
$d_1 + 2d_2 + 4d_3 = 5 $
\item For $A_{23}$: 
$d_1 + 2d_2 + 4d_3 = 6 $
\item For $A_{33}$: 
$d_1 + 2d_2 + 4d_3 = 8 $
\end{enumerate}
At the end we can gather this information and write up the maximum degree of these metric coefficients in matrix form.
\\
\\
\[ A_{ij} =
 \begin{pmatrix}
  \left[v_1^2,v_2\right] & \left[v_1^3,v_1v_2\right] & \left[v_1^5,v_1^3v_2,v_1v_2^2,v_3v_1\right]  \\
  \cdots & \left[v_1^4,v_1^2v_2,v_2^2, v_3\right] & \left[v_1^6,v_1^4v_2,v_1^2v_2^2,v_2^3,v_3v_1^2, v_3v_2\right]  \\
  \cdots  & \cdots  & \left[v_1^8,v_1^6v_2,v_1^4v_2^2,v_1^2v_2^3,v_2^4,v_3v_1^4,v_3v_1^2v_2,v_3v_2^2,v_3^2\right]   
 \end{pmatrix}
\]
 \\
 \\
For the matrix $a_{ij}$ and the vector $B_{ia}$ the result is: 
\[ a_{ij} =
 \begin{pmatrix} 1 & 0 & 0 \\ 
  \left[v_1\right] & 1 & 0  \\
  \left[v_1^3,v_1v_2\right] & \left[v_1^2,v_2\right] & 1 
 \end{pmatrix}, \quad B_{ia}=\begin{pmatrix} 
  \left[v_1\right] \\
  \left[v_1^2,v_2\right] \\
\left[v_1^4,v_1^2v_2,v_2^2,v_3\right]  
 \end{pmatrix}^T
\]
Note, however, we still have some freedom given in eq.(\ref{v-trafo}) to simplify the matrix $a_{ij}$. Using this transformation we can simplify $a_{ij}$ to be: 
\[ a_{ij} =
 \begin{pmatrix} 1 & 0 & 0 \\ 
  0 & 1 & 0  \\
  \left[v_1v_2\right] & \left[v_1^2\right] & 1 
 \end{pmatrix}.
\]

\subsection{General observations}
\label{General observations} 
A list of equalities for the matrix $A_{ij}$ and for our class of boost-vectors up to dimension $k=4$ can be found in appendix B. There are a few general observations. 
\paragraph{The case $k=1$: Kundt case.} We note that for $k=1$ there is only one possible non-trivial boost vector. All of these cases reduce to Kundt metrics for which $a_{ij}$ and $A_{ij}$ have only one component each: $a_{11}=1$, and $A_{11}=([v_1^2])$, and $B_{1a}=([v_1])$. These metrics are therefore consistent with the previous analysis of these metrics \cite{kundt}. 

\paragraph{The case $k=2$: Kundt or type II}  There are three non-trivial cases here, ${\mbold b}=(0,1), ~(1,1)$ and $(1,2)$. The first case can be considered as a $k=1$ by allowing  the transverse metric be pseudo-Riemannian (i.e., $g_{ab}dx^adx^b$ is pseudo-Riemannian). Hence, this is a Kundt case. For the cases $(1,1)$ and $(1,2)$, we note that for both we can reduce the matrix $a_{ij}=\delta_{ij}$. This is  the type II case in section \ref{sect:subclasses}. 
\\
\\
In 4 dimensions, this is the neutral case where there are no components $B_{ia}$. Consequently, in 4 dimensions both of these cases must also be Walker cases, type III. Hence, this is in agreement with the results found in \cite{VSI2}. 

\paragraph{Cases $(1,...,1)$: type II.} In all of the cases where the boost vector is $(1,...,1)$, all the $v_i$'s carry the same boost vector. This means that the matrix $a_{ij}=\delta_{ij}$, and hence of type II. Furthermore, the matrix $A_{ij}$ can at most be quadratic in $v_i$, and $B_{ia}$ at most linear in $v_i$. 
\\
\\
In the special case where the space is neutral of dimension $2k$, then $B_{ia}$ is not present and hence, the space is Walker (type III).

\paragraph{Covariantly constant ${\mbold F}$: type IV.} This case is a subclass of the Walker spaces. If we assume Walker, then this case is equivalent to the additional requirement: 
\[ \sum_{i}\partial_iA_{ij}=0. \]
The examples of  metrics obeying this condition are plentiful. An example is (all indices are written downstairs to avoid clutter):
\[
ds^2=2du_1(dv_1+v_2du_1)+2du_2(dv_2+v_3du_2)+2du_3(dv_3+v_1^8du_3).
\] 
The 3-form: 
\[ \boldsymbol{F}=du_1\wedge du_2\wedge du_3 \] 
is for this metric covariantly constant (and hence, is a null Killing-Yano tensor).\footnote{The fact that $\boldsymbol{F}$ satisfies Killing-Yano equations for this type of metrics has been spotted by using a systematic algorithm developed in \cite{KYsearch}.}
\subsection{The limiting spaces} 
As pointed out, applying the diffeomorphism $\varphi_\tau$ gives a space with identical invariants, including the limiting space as $\tau\rightarrow \infty$. With respect to a point $p$ which we can assume has coodinates $(v_i,u^i, x^a)=(0,u_0^i,x^a_0)$, the subleading powers of $v_i$ will tend to zero as $\tau\rightarrow \infty$. For example, if we consider the ${\mbold b}=(1,2,4)$ case, then choosing the $A_{23}$ component ($A_{23}\in \left[v_1^5,v_1^3v_2,v_1v_2^2,v_3v_1\right]$)
\[ \varphi^*_{\tau}A_{23}du^2du^3 \longrightarrow (av_1^5+bv_1^3v_2+cv_1v_2^2+dv_3v_1)du^2du^3,\]
where $a,b,c,d$ are functions of $(u^i_0,x^i)$. Note that in the limit, the coordinates $u^i_0$ tend towards a constant and in evaluating the invariants, one needs to keep $u^i_0$ fixed while the $x^a$ remains unaffected. Note also that the limit itself is symmetric w.r.t. $\varphi_\tau^*$ and thereby confirming Prop. \ref{prop:isometry}.  

\subsubsection{VSI spaces.} An interesting subclass of these spacetimes is the class where all polynomial curvature invariants vanish. Such spacetimes are those for which there exists a boost ${\mbold b}'$ so that the corresponding diffeomorphism $\varphi'_\tau$ has flat space as a limit: $ {\varphi'}_{\tau}^*g \longrightarrow \text{flat space}. $
This would contain all the spaces above where the polynomials are all subleading order. However, it also includes those spaces for which there exists a perturbation $\mbold \epsilon$ of ${\mbold b}$ so the boost ${\mbold b}'={\mbold b}+{\mbold\epsilon}$ gives flat space in the limit. 

Note also, that if there is a sequence of such boost limits which eventually leads to flat space, then this is sufficient to prove the space is VSI as well.  Therefore, a space is a VSI if there exists a sequence of boosts such that:
\[ g\longrightarrow g_1\longrightarrow \cdots \longrightarrow \text{flat space}.\]
Each arrow indicates an infinite boost limit.  

As an example of this is the following space:
\[
g=2du_1(dv_1+v_2du_1)+2du_2(dv_2+v_3du_2)+2du_3(dv_3+v_1^7du_3).
\] 
By using the boost ${\mbold b}_1=(1,2,4)$, the limiting space is:
\[ g_1=2du_1(dv_1+v_2du_1)+2du_2(dv_2+v_3du_2)+2du_3dv_3.\] 
Using next the boost ${\mbold b}_2=(1,1,0)$, gives the limit: 
\[ g_2=2du_1dv_1+2du_2dv_2+2du_3dv_3,\]
which is flat space. This proves that the  the metric $g$ is a VSI. 

\subsubsection{CSI spaces.} Another subclass of metrics are those that have polynomial invariants being all constants. Such spaces can be found by considering sequence of limits having a homogeneous space as an end product: 
\[ g\longrightarrow g_1\longrightarrow \cdots \longrightarrow \text{homogeneous space}.\]
All spaces having such a sequence of limits will be a CSI. 

For spaces where $k\geq 2$, a sequence of limits ending at a homogeneous space is sufficient but not necessary for it to be a CSI. For example, the 4-dimensional space (of boost-type (1,2))
\[ 2du_1(dv_1+v_2du_2)+2du_2(dv_2+v_1^4du_2),\] 
has only 3 Killing vectors, and hence, is not a homogeneous space. Yet, it is still a CSI space (and cannot be simplified further by a limiting procedure). One would like to have a criterion for a CSI space which states that a (degenerate) space is CSI if and only if there is a sequence: 
\[ g\longrightarrow g_1\longrightarrow \cdots \longrightarrow g_{\infty}. \]
However, it is not clear what the metric of $g_{\infty}$ is, but there are some conditions that it has to satisfy. Firstly, the transverse space needs to be a homogeneous space. Second, the additional null-directions will possess (at least) $k$ translations and one boost as isometries: hence,  $g_{\infty}$ will possess a minimum of $(k+m+1)$ Killing vectors (a homogeneous space has at least $(2k+m)$ Killing vectors, hence, $g_\infty$ need not be homogeneous). 

Examples of such CSI spaces are spaces having a Lie group $G$ as a transversal space equipped with a left-invariant metric. If ${\mbold \omega}^a$ are left-invariant 1-forms on the group $G$, then $g_{ab}{\mbold \omega}^a{\mbold \omega}^b$, where $g_{ab}$ is a constant matrix, is a left-invariant metric on $G$. Furthermore, if the leading order coefficients of the polynomials in the $v_i$'s (saturating the limits in ${\mbold V}\cdot {\mbold b}=0$), in $a_{ij}$ and $A_{ij}$ are constants, as well as the coefficients of the matrix $B_{ia}{\mbold \omega}^a$ (using the left-invariant 1-forms as basis 1-forms), then it is a CSI. 

As an 8-dimensional example, let 
\[ {\mbold\omega}^1=dw, \quad {\mbold \omega}^2=e^{2w}[dx +(1/2)(ydz-zdy)],\quad {\mbold \omega}^3 = e^wdy, \quad {\mbold \omega}^4=e^wdz, \] 
(or any left-invariant one-forms on a 4-dimensional Lie group)
and $g_{ab}$ be any constant matrix. Let also 
${\mbold\eta}$ be any (constant) linear combination of the ${\mbold \omega}^a$'s, i.e., 
\[ {\mbold \eta}=a{\mbold \omega}^1+b{\mbold \omega}^2+c{\mbold \omega}^3+d{\mbold \omega}^4.\]
So an example of an 8-dimensional CSI is:
\beq
g&=&
2du_1[dv_1+(a_1v_2+a_2(u_i,x_a))du_1+(b_1v_1+b_2(u_i,x_a)){\mbold\eta}_1]\nonumber \\ & +&2du_2[dv_2+(c_1v^4_1+c_2(u_i,x_a)v^3_1)du_2+(d_1v_1^2+d_2(u_i,x_a)v_1+d_3(u_i,x_a)){\mbold\eta}_2]\nonumber \\
&+& g_{ab}{\mbold\omega}^a{\mbold\omega}^b,
\eeq
where all are constants except when $(u_i,x_a)$-dependence is explicitly mentioned (more non-zero metric functions are possible though, only some are included here). In this case the limit is: 
\beq
2du_1[dv_1+a_1v_2du_1+b_1v_1{\mbold\eta}_1]+2du_2[dv_2+c_1v^4_1du_2+d_1v_1^2{\mbold\eta}_2]\nonumber 
+ g_{ab}{\mbold\omega}^a{\mbold\omega}^b,
\eeq
which is a homogeneous space. A plethora of other examples of CSI spaces can be found using the same procedure. 

\section{Conclusion}
In this paper we have discussed pseudo-Riemannian spaces with degenerate curvature stucture. Under a simple assumption we found a class of metrics being $\mathcal{I}$-degenerate, eq. (\ref{result}). This class includes all known examples to date, as well as new families of examples showing that this class is bigger than previously known. 

Examples of VSI and CSI spaces have been given, as well as metrics with more special curvature properties. For example, contained in the class, there are the metrics of type IV which possess a covariantly constant null $k$-form $F$. Clearly, there are also subtypes of the main types listed here which are amenable for further study. 

A  question is still lingering: are these all such $\mathcal{I}$-degenerate spaces? This question depends on the following  crucial assumption:
{\it For any point $p\in U$, there exists a (non-trivial) boost so the metric, in the coordinates given, has a finite limit $\lim_{\tau\rightarrow\infty}\varphi_\tau^*g\in \mf{M}$.
}
From invariant theory, we know that such a limit exists point-wise \cite{RS,VSI2}, however, the extension to a neighbourhood, $U$, is unsettled. 

\section*{Acknowledgement}
KY would like to thank Tsuyoshi Houri for helpful discussions and for making his work available to us before publication, which eventually led us to the class of spaces we have considered in this article. The same author is also grateful to University of Stavanger for the support. This work was partly supported by the JSPS Grant-in-Aid for Scientific Research No. 26$\cdot $1204.

\appendix
\section{Connection 1-forms}

We start from the canonical frame $\{ \boldsymbol{e}^{\mu } \}  = \{ \boldsymbol{\ell }^i , \boldsymbol{m}^a , \boldsymbol{n}^{\hat{i}} \}$ constructed in section 2. The metric is written as
\begin{equation}
ds^2 = g_{\mu \nu } \boldsymbol{e}^{\mu } \boldsymbol{e}^{\nu } = 2\delta _{i \hat{j}} \boldsymbol{\ell }^i \boldsymbol{n}^{\hat{j}} + \eta _{ab} \boldsymbol{m}^a \boldsymbol{m}^b \ .
\end{equation}
The components of connection 1-forms in this frame
\begin{equation}
\omega ^{\lambda }_{\ \mu \nu } \boldsymbol{e}^{\nu } = \nabla _{\boldsymbol{e}_{\mu }} \boldsymbol{e}^{\lambda }
\end{equation}
is related to the exterior derivatives of the canonical frame
\begin{equation}
d\boldsymbol{e}^{\lambda } = -\frac{1}{2} C^{\lambda }_{\ \mu \nu } \boldsymbol{e}^{\mu } \wedge \boldsymbol{e}^{\nu } 
\end{equation}
by
\begin{equation}
\omega _{\lambda \mu \nu } = g_{\lambda \rho }\omega ^{\rho }_{\ \mu \nu }= \frac{1}{2} \left( g_{\mu \rho }C^{\rho }_{\ \lambda \nu } + g_{\nu \rho } C^{\rho }_{\ \lambda \mu } - g_{\lambda \rho }C^{\rho }_{\ \mu \nu } \right) \ .
\end{equation}
For instance, 
\begin{equation}
\begin{split}
& \omega ^k_{\ \hat{i} \hat{j}} = \omega _{\hat{k}\hat{i}\hat{j}} = \frac{1}{2} \left( C^i_{\ \hat{k} \hat{j}} + C^j_{\ \hat{k}\hat{i}} - C^k_{\ \hat{i} \hat{j}} \right) \ , \\
& \omega ^k_{\ \hat{i}a} = \omega _{\hat{k}\hat{i}a} = \frac{1}{2}\left(  C^i_{\ \hat{k} a} \pm  C^a_{\ \hat{k}\hat{i}} - C^k_{\ a \hat{i}} \right) \ .
\end{split}
\end{equation}
Hence from the condition $d\boldsymbol{\ell }^i = \boldsymbol{\sigma }^i_{\ j} \wedge \boldsymbol{\ell }^j $ alone, one can derive
\begin{equation}
 \omega ^k_{\ \hat{i} \hat{j}} = 0 \ , \quad \omega ^k_{\ \hat{i}a} = \pm \frac{1}{2}C^a_{ \ \hat{k}\hat{i}}
 \end{equation}
 and see that
 \begin{equation}
 \omega ^j_{\ \hat{i}a} - \omega ^i_{\ \hat{j}a} = 0 \quad \Leftrightarrow \quad \omega ^i_{\ \hat{j}a} = 0 \ .
 \end{equation}
 
Let us now be more specific and take the metric and coordinate 
\begin{equation}
ds^2 = 2du^{ i  } \left( a_{ i   j  }dv^{ j  } + A_{ i   j  }du^{ j  } + B_{ i  a} dx^a \right) + g_{ab} dx^a dx^b \ .
\end{equation}
As the canonical frame, one can choose
\begin{equation}
\begin{split}
& \boldsymbol{\ell }^{ i  } = du^{ i  } \ , \quad \boldsymbol{m}^a = e^a_{\ b} dx^b \ , \quad g_{ab} = \eta _{cd} e^c_{\ a} e^d_{\ b} \ , \\
& \boldsymbol{n}^{\hat{ i  }} = a_{ i   j  } dv^{ j  } + A_{ i   j  } du^{ j  } + B_{ i  a}dx^a \ . 
\end{split}
\end{equation}
Exterior derivatives now satisfy
\begin{equation}
d\boldsymbol{\ell }^{ i  } = 0 \  , \quad d\boldsymbol{n}^{\hat{ i  }} = -\frac{1}{2} C^{\hat{ i  }}_{\ \mu\nu}\boldsymbol{e}^\mu \wedge \boldsymbol{e}^\nu \ , \quad d\boldsymbol{m}^a = -\frac{1}{2}C^a_{~\mu b} \boldsymbol{e}^\mu \wedge \boldsymbol{m}^b \ .
\end{equation}
In the present setting, one can classify the connection components into 18 groups. 
\begin{equation}
\begin{split}
\omega _{ i   j   l  } = & \frac{1}{2} \left( g_{ j  \hat{ j  }}C^{\hat{ j  }}_{\  i   l  } + g_{ l  \hat{ l  }}C^{\hat{ l  }}_{\  i   j  } - g_{ i  \hat{ i  }}C^{\hat{ i  }}_{\  j   l  } \right) \ , \\
\omega _{ i  \hat{ j  } l  } = & \frac{1}{2} \left( g_{ l  \hat{ l  }} C^{\hat{ l  }}_{\  i  \hat{ j  } } - g_{ i  \hat{ i  }}C^{\hat{ i  }}_{\ \hat{ j  }  l  } \right) \ , \quad \omega _{ i   j  \hat{ l  }} =  \frac{1}{2} \left( g_{ j  \hat{ j  }}C^{\hat{ j  }}_{\  i  \hat{ l  }} - g_{ i  \hat{ i  }}C^{\hat{ i  }}_{\  j  \hat{ l  }} \right) \ , \\
\omega _{ i  a  l  } = & \frac{1}{2}\left( g_{ l  \hat{ l  }}C^{\hat{ l  }}_{\  i  a} - g_{ i  \hat{ i  }}C^{\hat{ i  }}_{\ a l  } \right) \ , \quad \omega _{ i   j  a} = \frac{1}{2} \left( g_{ j  \hat{ j  }}C^{\hat{ j  }}_{\  i  a} - g_{ i  \hat{ i  }}C^{\hat{ i  }}_{\  j  a} \right)  \ , \\
\omega _{\hat{ i  }\hat{ j  } l  } = & \frac{1}{2} g_{ l  \hat{ l  }}C^{\hat{ l  }}_{\ \hat{ i  } \hat{ j  }} = - \omega _{ l  \hat{ i  } \hat{ j  }}  \ , \\
\omega _{\hat{ i  }a l  } = & \frac{1}{2} g_{ l  \hat{ l  }} C^{\hat{ l  }}_{\ \hat{ i  }a} = \omega _{ l  a\hat{ i  }} = -\omega _{ l  \hat{ i  }a} \  , \\
\omega _{ab l  } = & \frac{1}{2} \left( g_{bb}C^b_{\ a l  } + g_{ l  \hat{ l  }}C^{\hat{ l  }}_{\ ab} - g_{aa}C^a_{\ b l  } \right) \ , \\
\omega _{\hat{ i  }\hat{ j  } \hat{ l  }} = & \omega _{\hat{ i  }a\hat{ l  }} = \omega _{\hat{ i  }\hat{ j  }a} = 0 \ , \\
\omega _{ab\hat{ l  }} = & \frac{1}{2} \left( g_{bb}C^b_{\ a\hat{ l  }} - g_{aa} C^a_{\ b\hat{ l  }} \right) \ , \quad \omega _{\hat{ i  }ab} = \frac{1}{2} \left( g_{aa}C^a_{\ \hat{ i  }b}+ g_{bb}C^b_{\ \hat{ i  }a} \right) \ , \\
\omega _{ i  ab} = & \frac{1}{2} \left( g_{aa}C^a_{\  i  b} + g_{bb}C^b_{\  i  a} - g_{ i  \hat{ i  }}C^{\hat{ i  }}_{\ ab} \right) \  , \\
\omega _{abc} = & \frac{1}{2} \left( g_{bb}C^b_{\ ac} + g_{cc}C^c_{\ ab} - g_{aa}C^a_{\ bc} \right) \ . 
\end{split}
\end{equation}

We are mostly interested in covariant derivatives of $\boldsymbol{\ell }^{ i  }$ that are characterised by $\omega _{\hat{ i}  jk}$.
In particular, for the null k-form $\boldsymbol{F}$, we have
\begin{equation}
\begin{split}
\nabla _{\mu } \boldsymbol{F} = & \nabla _\mu \boldsymbol{\ell }^1 \wedge \boldsymbol{\ell }^2 \wedge \cdots \wedge \boldsymbol{\ell }^k \\
= &  \omega ^{ i  }_{\ \mu  i  } \boldsymbol{\ell }^1 \wedge \boldsymbol{\ell }^2 \wedge \cdots \wedge \boldsymbol{\ell }^k \\
+& \omega ^1_{\ \mu \hat{ i  }}\boldsymbol{n}^{\hat{ i  }}\wedge \boldsymbol{\ell }^2  \wedge \cdots \wedge \boldsymbol{\ell }^k + \cdots + (-1)^{k-1}\omega ^k_{\ \mu \hat{ i  }}\boldsymbol{n}^{\hat{ i  }} \wedge \boldsymbol{\ell }^1 \wedge \cdots \wedge \boldsymbol{\ell }^{k-1} \\
+&  \omega ^1_{\ \mu a}\boldsymbol{m}^{a}\wedge \boldsymbol{\ell }^2  \wedge \cdots \wedge \boldsymbol{\ell }^k + \cdots + (-1)^{k-1}\omega ^k_{\ \mu a}\boldsymbol{m}^{a} \wedge \boldsymbol{\ell }^1 \wedge \cdots \wedge \boldsymbol{\ell }^{k-1} \ .
\end{split}
\end{equation}
Therefore, $\nabla \boldsymbol{F}$ contains all the information for $\omega ^i_{\ \mu \hat{j}}$ and $\omega ^i_{\ \mu a}$. 
Among the connection coefficients appearing in the above formula, some of them are already zero (by the surface-forming conditions \ref{cond:1} and \ref{cond:2}).
The relevant non-vanishing ones are
\begin{equation}
\begin{split}
\omega _{\hat{ j  } i   j  } = & -\frac{1}{2} g_{ j  \hat{ j  }}C^{\hat{ j  }}_{\  i  \hat{ j  }} \ , \\
\omega _{\hat{ i  } j  \hat{ l  }} = & \frac{1}{2}g_{ j  \hat{ j  }}C^{\hat{ j  }}_{\ \hat{ i  }\hat{ l  }} \ , \\
\omega _{\hat{ i  } j  a} = & \frac{1}{2}g_{ j  \hat{ j  }}C^{\hat{ j  }}_{\ \hat{ i  }a} \ , \\
\omega _{\hat{ i  }ab} = & \frac{1}{2} \left( g_{aa}C^a_{\ \hat{ i  }b} + g_{bb}C^b_{\ \hat{ i  }a} \right) \ .
\end{split}
\end{equation}
To express them in terms of the metric components, one need compute $d\boldsymbol{n}^{\hat{ i  }}$ and $d\boldsymbol{m}^a$. 
Noting that
\begin{equation}
\begin{split}
dx^a = & \left( e^{-1} \right) ^a_{\ b}\boldsymbol{m}^b \ , \\
dv^{ j  } = & \left( a^{-1} \right) _{ j   l  } \boldsymbol{n}^{\hat{ l  }} - \tilde{A}_{ j   l  }\boldsymbol{\ell }^{ l  } - \tilde{B}_{ i  a}\left( e^{-1}\right) ^a_{\ b}\boldsymbol{m}^b \ , 
\end{split}
\end{equation}
where
\begin{equation}
\tilde{A}_{ i   j  } = \left( a^{-1}\right) _{ i   l  }A_{ l   j  } \ , \quad \tilde{B}_{ i  a} = \left( a^{-1} \right) _{ i   j  }B_{ j  a} \ , 
\end{equation}
we obtain
\begin{equation}
\begin{split}
d\boldsymbol{n}^{\hat{ i  }} = & \left[ \partial _{u^{ j  }}A_{ i   l  } - \left( \partial _{u^{ j  }}a_{ i   n  } - \partial _{v^{ n  }}A_{ i   j  }\right) \tilde{A}_{ n   l  } + \left( \partial _{v^{ n  }}a_{ i   m  } \right)\tilde{A}_{ n   j  } \tilde{A}_{ m   l  } \right]  \boldsymbol{\ell }^{ j  }\wedge \boldsymbol{\ell }^{ l  } \\
&+ \Big[ \left( \partial _{u^{ j  }}a_{ i   l  } - \partial _{v^{ l  }}A_{ i   j  } \right) \left( a^{-1} \right) _{ l   n  } \\
& \quad + \left( \left( a^{-1}\right) _{ m   n  }\tilde{A}_{ l   j  } - \left( a^{-1}\right) _{ l   n  }\tilde{A}_{ m   j  }\right) \partial _{v^{ m  }}a_{ i   l  } \Big] \boldsymbol{\ell }^{ j  }\wedge \boldsymbol{n}^{\hat{ n  }} \\
& + \Big[ \partial _{u^{ j  }}B_{ i  a}-\left( \partial _{x^a} - \tilde{B}_{ l  a}\partial _{v^{ l  }} \right) A_{ i   j  } \\
& \quad  -\left( \tilde{B}_{ l  a}\left( \partial _{u^{ j  }} - \tilde{A}_{ n   j  }\partial _{v^{ n  }} \right) + \tilde{A}_{ l   j  }\tilde{B}_{ n  a}\partial _{v^{ n  }} \right) a_{ i   l  } \Big] \left( e^{-1}\right) ^a_{\ b}\boldsymbol{\ell }^{ j  }\wedge \boldsymbol{m}^b \\
& + \left( \partial _{v^{ j  }}a_{ i   l  } \right) \left( a^{-1}\right) _{ j   n  } \left( a^{-1} \right) _{ l   m  } \boldsymbol{n}^{\hat{ n  }}\wedge \boldsymbol{n}^{\hat{ m  }} \\
& + \Big[ \left( \partial _{v^{ j  }}B_{ i  b} - \partial _{x^b}a_{ i   j  } \right) \left( a^{-1}\right) _{ j   n  } \\
& \quad - \left( \partial _{v^{ j  }}a_{ i   l  } \right) \left( \left( a^{-1}\right) _{ j   n  }\tilde{B}_{ l  b} - \left( a^{-1}\right)_{ l   n  }\tilde{B}_{ j  b} \right) \Big] \left( e^{-1}\right) ^b_{\ a} \boldsymbol{n}^{\hat{ n  }}\wedge \boldsymbol{m}^a  \\
& + \Big[ \tilde{B}_{ j  c} \left( \tilde{B}_{ l  d} \partial _{v^{ j  }}a_{ i   l  } + \partial _{x^d}a_{ i   j  } \right) \\
& \quad + \left( \partial _{x^c} - \tilde{B}_{ j  b}\partial _{v^{ j  }} \right) B_{ i  d} \Big] \left( e^{-1}\right) ^c_{\ a} \left( e^{-1} \right) ^d_{\ b} \boldsymbol{m}^a \wedge \boldsymbol{m}^b \ , 
\end{split}
\end{equation}
\begin{equation}
\begin{split}
d\boldsymbol{m}^a = & \left( \partial _{u^{ i  }} - \tilde{A}_{ j   i  } \partial _{v^{ j  }} \right) e^a_{\ c} \left( e^{-1} \right) ^c_{\ b} \boldsymbol{\ell }^{ i  } \wedge \boldsymbol{m}^b \\
& + \left( \partial _{v^{ i  }}e^a_{\ c} \right) \left( a^{-1}\right) _{ i   j  } \left( e^{-1} \right) ^c_{\ f} \boldsymbol{n}^{\hat{ j  }} \wedge \boldsymbol{m}^b \\
& + \left[ \left( \partial _{x^b} - \tilde{B}_{ i  b} \partial _{v^{ i  }} \right) e^a_{\ c} \right] \left( e^{-1} \right) ^b_{ \ d}\left( e^{-1} \right) ^c_{\ f} \boldsymbol{m}^d \wedge \boldsymbol{m}^f \ .
\end{split}
\end{equation}
From these expression, one can readily read off the correspondence between the conditions on the covariant derivatives of $\boldsymbol{F}$ and the restrictions on the metric components given in section 2.1.
\section{The metric components $A_{ij}$}
We set $\boldsymbol{V} \cdot \boldsymbol{b} = 0$. The indices $i$ and $j$ run from $1,2,...,k$. For each value of ($i$,$j$) we get an equality. We read off the maximum of the $d_i's$ from this. The dimension $k$ of the $(k, k + m)$ manifold is varied from 1 to 4. These results are easily manipulated to give the other differentials.
\subsection{The case \underline{$k = 1$}}
Possible boost vectors:
\begin{enumerate}
\item 0
\item 1
\end{enumerate}
$\textbf{V} \cdot \boldsymbol{b} = 0$ gives, in each case:
\begin{enumerate}
\item Non-degenerate
\item $~\forall$ $A_{ij}$: $d_1 = 2$, Kundt , see subsection \ref{General observations}
\end{enumerate}
\subsection{The case \underline{$k = 2$}}
Possible boost vectors:
\begin{enumerate}
\item (0,0)
\item (0,1)
\item (1,1)
\item (1,2)
\end{enumerate}
$\textbf{V} \cdot \boldsymbol{b} = 0$ gives, in each case:
\begin{enumerate}
\item Non-degenerate
\item $~\forall$ $A_{ij}$: $d_2 = 2$ Kundt, see subsection \ref{General observations}
\item $~\forall$ $A_{ij}$: $d_1 + d_2 = 2$ $\rightarrow A_{ij} \in [v_1^2,v_1v_2,v_2^2]$
\item \begin{enumerate}[i)]
\item For $A_{11}$: $d_1 + 2 d_2 = 2 \rightarrow $ $A_{11} \in [v_1^2,v_2]$ 
\item For $A_{12}$: $d_1 + 2 d_2 = 3 \rightarrow$ $A_{12} \in [v_1^3,v_1v_2]$
\item For $A_{22}$: $d_1 + 2 d_2 = 4 \rightarrow$ $A_{22} \in [v_1^4,v_1^2v_2,v_2^2]$  
\end{enumerate} 
\end{enumerate}
Before we go on we want to make two comments that simplify the notation.
\\
\\
Firstly, we forego writing the "$A_{22} \in [v_1^4,v_1^2v_2,v_2^2]$" part of the expression. It contains the same information as its corresponding equality. Secondly, for the first half of the boost vectors, we have the same situation as for \underline{\textit{k = 1}}. The difference is that results previously obtained for $d_1$ now pertain to $d_2$. In conclusion, each time we increase the value of \textit{k}, the first half of the boost vectors, the ones with a zero in the first slot, will give the same results as the situation where \textit{k} is one less. The difference is that the result for $v_i$ is shifted to $v_{i+1}$. Henceforth, we will therefore omit these. 

\subsection{The case \underline{$k = 3$}}
Possible boost vectors:
\begin{enumerate}
\setcounter{enumi}{4}
\item (1,1,1)
\item (1,1,2)
\item (1,2,2)
\item (1,2,4)
\end{enumerate}
$\textbf{V} \cdot \boldsymbol{b} = 0$ gives, in each case:
\begin{enumerate}
\setcounter{enumi}{4}
\item $~\forall$ $A_{ij}$: \underline{$d_1 + d_2 + d_3 = 2$} 
\item \begin{enumerate}[i)]
\item For $A_{11}$: \underline{$d_1 + d_2 + 2d_3 = 2$} 
\item For $A_{13},A_{23}$: \underline{$d_1 + d_2 + 2d_3 = 3$} 
\item For $A_{33}$: \underline{$d_1 + d_2 + 2d_3 = 4$} 
\end{enumerate}
\item \begin{enumerate}[i)]
\item For $A_{11}$: \underline{$d_1 + 2d_2 + 2d_3 = 2$}
\item For $A_{12},A_{23}$: \underline{$d_1 + 2d_2 + 2d_3 = 3$}
\item For $A_{33}$: \underline{$d_1 + 2d_2 + 2d_3 = 4$}
\end{enumerate}
\item \begin{enumerate}[i)]
\item For $A_{11}$: \underline{$d_1 + 2d_2 + 4d_3 = 2$}
\item For $A_{12}$: \underline{$d_1 + 2d_2 + 4d_3 = 3$}
\item For $A_{22}$: \underline{$d_1 + 2d_2 + 4d_3 = 4$}
\item For $A_{13}$: \underline{$d_1 + 2d_2 + 4d_3 = 5$}
\item For $A_{23}$: \underline{$d_1 + 2d_2 + 4d_3 = 6$}
\item For $A_{33}$: \underline{$d_1 + 2d_2 + 4d_3 = 8$}
\end{enumerate}
\end{enumerate}
\subsection{The case \underline{$k = 4$}}
Possible boost vectors:
\begin{enumerate}
\setcounter{enumi}{8}
\item (1,1,1,1)
\item (1,1,1,2)
\item (1,1,2,2)
\item (1,1,2,4)
\item (1,2,2,2)
\item (1,2,2,4)
\item (1,2,4,4)
\item (1,2,4,8)
\end{enumerate}
$\textbf{V} \cdot \boldsymbol{b} = 0$ gives
\begin{enumerate}
\setcounter{enumi}{8}
\item $~\forall$ $A_{ij}$: $d_1 + d_2 + d_3 + d_4 = 2$ 
\item \begin{enumerate}[i)]
\item  For $A_{11},A_{22},A_{33},A_{12},A_{13},A_{23}$: \underline{$d_1 + d_2 + d_3 + 2d_4 = 2$}
\item For $A_{14},A_{24},A_{34}$: \underline{$d_1 + d_2 + d_3 + 2d_4 = 3$}
\item For $A_{44}$: \underline{$d_1 + d_2 + d_3 + 2d_4 = 4$}
\end{enumerate}
\item \begin{enumerate}[i)]
\item For $A_{11},A_{22},A_{12}$: \underline{$d_1 + d_2 + 2d_3 + 2d_4 = 2$}
\item For $A_{13},A_{14},A_{23},A_{24}$: \underline{$d_1 + d_2 + 2d_3 + 2d_4 = 3$}
\item For $A_{33},A_{44},A_{34}$: \underline{$d_1 + d_2 + 2d_3 + 2d_4 = 4$}
\end{enumerate}
\item \begin{enumerate}[i)]
\item For $A_{11},A_{22},A_{12}$: \underline{$d_1 + d_2 + 2d_3 + 4d_4 = 2$}
\item For $A_{13},A_{23}$: \underline{$d_1 + d_2 + 2d_3 + 4d_4 = 3$}
\item For $A_{33}$: \underline{$d_1 + d_2 + 2d_3 + 4d_4 = 4$}
\item For $A_{14},A_{24}$: \underline{$d_1 + d_2 + 2d_3 + 4d_4 = 5$}
\item For $A_{34}$: \underline{$d_1 + d_2 + 2d_3 + 4d_4 = 6$}
\item For $A_{44}$: \underline{$d_1 + d_2 + 2d_3 + 4d_4 = 8$}
\end{enumerate}
\item \begin{enumerate}[i)]
\item For $A_{11}$: \underline{$d_1 + 2d_2 + 2d_3 + 2d_4 = 2$}
\item For $A_{12},A_{13},A_{14}$: \underline{$d_1 + 2d_2 + 2d_3 + 2d_4 = 3$}
\item For $A_{22},A_{33},A_{44},A_{23},A_{24},A_{34}$: \underline{$d_1 + 2d_2 + 2d_3 + 2d_4 = 4$}
\end{enumerate}
\item \begin{enumerate}[i)]
\item For $A_{11}$: \underline{$d_1 + 2d_2 + 2d_3 + 4d_4 = 2$}
\item For $A_{12},A_{13}$: \underline{$d_1 + 2d_2 + 2d_3 + 4d_4 = 3$}
\item For $A_{22},A_{33},A_{23}$: \underline{$d_1 + 2d_2 + 2d_3 + 4d_4 = 4$}
\item For $A_{14}$: \underline{$d_1 + 2d_2 + 2d_3 + 4d_4 = 5$}
\item For $A_{24},A_{34}$: \underline{$d_1 + 2d_2 + 2d_3 + 4d_4 = 6$}
\item For $A_{44}$: \underline{$d_1 + 2d_2 + 2d_3 + 4d_4 = 8$}
\end{enumerate}
\item \begin{enumerate}[i)]
\item For $A_{11}$: \underline{$d_1 + 2d_2 + 4d_3 + 4d_4 = 2$}
\item For $A_{12}$: \underline{$d_1 + 2d_2 + 4d_3 + 4d_4 = 3$}
\item For $A_{22}$: \underline{$d_1 + 2d_2 + 4d_3 + 4d_4 = 4$}
\item For $A_{13},A_{14}$: \underline{$d_1 + 2d_2 + 4d_3 + 4d_4 = 5$}
\item For $A_{23},A_{24}$: \underline{$d_1 + 2d_2 + 4d_3 + 4d_4 = 6$}
\item For $A_{33},A_{44},A_{34}$: \underline{$d_1 + 2d_2 + 4d_3 + 4d_4 = 8$}
\end{enumerate}
\item \begin{enumerate}[i)]
\item For $A_{11}$: \underline{$d_1 + 2d_2 + 4d_3 + 8d_4 = 2$}
\item For $A_{12}$: \underline{$d_1 + 2d_2 + 4d_3 + 8d_4 = 3$}
\item For $A_{22}$: \underline{$d_1 + 2d_2 + 4d_3 + 8d_4 = 4$}
\item For $A_{13}$: \underline{$d_1 + 2d_2 + 4d_3 + 8d_4 = 5$}
\item For $A_{23}$: \underline{$d_1 + 2d_2 + 4d_3 + 8d_4 = 6$}
\item For $A_{33}$: \underline{$d_1 + 2d_2 + 4d_3 + 8d_4 = 8$}
\item For $A_{14}$: \underline{$d_1 + 2d_2 + 4d_3 + 8d_4 = 9$}
\item For $A_{24}$: \underline{$d_1 + 2d_2 + 4d_3 + 8d_4 = 10$}
\item For $A_{34}$: \underline{$d_1 + 2d_2 + 4d_3 + 8d_4 = 12$}
\item For $A_{44}$: \underline{$d_1 + 2d_2 + 4d_3 + 8d_4 = 16$}
\end{enumerate}
\end{enumerate}


\begin{thebibliography}{abc}

\bibitem{invariants} 
  A.~Coley, S.~Hervik and N.~Pelavas,
  Class.\ Quant.\ Grav.\  {\bf 27}, 102001 (2010)
  [arXiv:1003.2373 [gr-qc]].
  
\bibitem{Goodman}
R. Goodman and N.R. Wallach, \textit{Symmetry, Representations and Invariants}, Springer, 2009.

\bibitem{degen}
A. Coley, S. Hervik and N.    
Pelavas, 2009, Class. Quant. Grav. {\bf 26}, 025013    
[arXiv:0901.0791];

\bibitem{OP} S Hervik and A. Coley, 2010,  
Class. Quant. Grav. {\bf 27}, 095014  [arXiv:1002.0505];  
A. Coley and  S. Hervik, 
2009, Class. Quant. Grav. {\bf 27}, 
015002  [arXiv:0909.1160]. 

\bibitem{VSI} V. Pravda, A. Pravdova, A. Coley and R. Milson, 2002 
Class. Quant. Grav. 
{\bf{19}}, 6213 [gr-qc/0209024];  
A. Coley, R. Milson, V. Pravda, A. Pravdova, 2004, 
Class. Quant. Grav. {\bf{21}}, 5519;   
A. Coley, A. Fuster, S. Hervik and N. Pelavas, 2006, 
Class. Quant. Grav.  
{\bf{23}}, 7431.

\bibitem{CSI}
A. Coley, S. Hervik and N.    
Pelavas, 2006, Class. Quant. Grav. {\bf 23}, 3053    
[arXiv:gr-qc/0509113]  

\bibitem{kundt}   
A. Coley, S. Hervik, G. Papadopoulos and N. Pelavas, 2009, Class. Quant. Grav. {\bf 26}, 105016    
[arXiv:0901.0394];  

\bibitem{Walker} 
A.G. Walker, 1949, Quart. J. Math. (Oxford), {\bf 20}, 135-45; 
A.G. Walker, 1950, Quart. J. Math. (Oxford) (2), {\bf 1}, 69-79  

\bibitem{VSI1} 
S. Hervik,  and A. Coley,
Class.Quant.Grav. {\bf 28}, 015008 (2011)
\texttt{[arXiv:1008.2838]} 

\bibitem{Alcolado} 
  A.~Alcolado, A.~MacDougall, A.~Coley and S.~Hervik,
  J.\ Geom.\ Phys.\  {\bf 62}, 594 (2012)
  [arXiv:1104.3799 [math-ph]]. 

\bibitem{VSI2}
S.Hervik, 
Class.Quant.Grav.{\bf 29}:095011 (2012)
\\
\url{http://iopscience.iop.org/0264-9381/29/9/095011}  


\bibitem{twistor}
R. Penrose and W. Rindler, \emph{Twistors and Space-time}, Cambridge Uni. Press, volumes 1 \& 2 (1986); \\
M. Dunajski, 	\textit{J. Phys.} {\bf A42} (2009) 404004.

\bibitem{Dun}
M. Dunajski, \emph{Proc. Roy. Soc. Lond.} {\bf A458} (2002) 1205-1222;  \\
M. Dunajski and P. Tod, \emph{Math. Proc. Cam. Phil. Soc.} {\bf 148} (2010) 485, \texttt{arxiv:0901.2261v3 [math.DG]};\\
A.S. Galaev, \texttt{arXiv:1002.2064v1 [math.DG]}

  

\bibitem{RS}
R.W. Richardson and P.J. Slodowy, 1990, J. London Math. Soc. (2) {\bf 42}: 409-429.



\bibitem{align}
 S.~Hervik,
  Class.\ Quant.\ Grav.\  {\bf 28}, 215009 (2011)
  [arXiv:1109.2551 [gr-qc]].

\bibitem{degenInv}
  S.~Hervik, M.~Ortaggio and L.~Wylleman,
  Class.\ Quant.\ Grav.\  {\bf 30}, 165014 (2013)
  [arXiv:1203.3563 [gr-qc]].  

\bibitem{bw} 
S. Hervik and A. Coley, 
Int.J.Geom.Meth.Mod.Phys. {\bf 08}, 1679-1685 (2011)
\texttt{[arXiv:1008.3021]}.

\bibitem{KYsearch}
T. Houri and Y. Yasui, 2014,  
Class.Quant.Grav.32:055002 (2015)
\texttt{[arXiv:1410.1023]}.
  


\end{thebibliography}
\end{document}